\documentclass{article}
\usepackage{arxiv}
\usepackage{amsthm}
\usepackage{amsmath}
\usepackage[T1]{fontenc}
\usepackage{hyperref}
\usepackage{url}
\usepackage{booktabs}
\usepackage{amsfonts}
\usepackage{nicefrac}
\usepackage{microtype}
\usepackage{cleveref}
\usepackage{graphicx}
\usepackage{tikz}
\usetikzlibrary{positioning,calc,decorations.pathreplacing,patterns,arrows.meta,shapes,backgrounds}
\usepackage{natbib}
\usepackage{doi}
\usepackage{xcolor}
\definecolor{darkgreen}{RGB}{0,80,0}
\usepackage{pgfplots}
\pgfplotsset{compat=1.18}
\usepackage{listings}
\usepackage[many]{tcolorbox}  
\usepackage{fontspec}
\usepackage{xparse}
\usepackage{expl3}
\usepackage{tikz}
\usepackage{array}

\definecolor{jnoun}{RGB}{247, 194, 214}
\definecolor{jverb}{RGB}{217 211 255}
\definecolor{jconjunction}{RGB}{173 242 211}
\definecolor{jadverb}{RGB}{195, 239, 252}
\definecolor{jcontrol}{RGB}{255 187 174}
\definecolor{jother}{gray}{0.92}

\tcbset{jstyle/.style={
  enhanced,
  nobeforeafter,
  tcbox raise base,
  boxrule=0.4pt,
  top=0.0mm,
  bottom=0.0mm,
  right=0mm,
  left=0mm,
  arc=2pt,
  boxsep=2pt,
  before upper={\vphantom{[,}},
  colframe=black!50!white,
  fontupper=\ttfamily,
  fontlower=\ttfamily,
}}

\newcommand{\bo}{\text{\ttfamily (}}
\newcommand{\bc}{\text{\ttfamily )}}
\newcommand{\jv}[2][s]{%
  \ifstrequal{#1}{r}%
    {\tcbox[jstyle, colback=jverb]{#2}}%
    {\tcbox[jstyle, colback=jverb, sharp corners]{#2}}%
}
\newcommand{\jn}[1]{%
  \tcbox[jstyle, colback=jnoun, sharp corners]{#1}%
}
\newcommand{\ja}[1]{%
  \tcbox[jstyle, colback=jadverb, sharp corners]{#1}%
}
\newcommand{\jc}[1]{%
  \tcbox[jstyle, colback=jconjunction, sharp corners]{#1}%
}
\newcommand{\jo}[1]{%
  \tcbox[jstyle, colback=jother, sharp corners]{#1}%
}
\newcommand{\jt}[1]{%
  \tcbox[jstyle, colback=jcontrol, sharp corners]{#1}%
}

\tikzset{
  jcell/.style={
    draw=black!50!white,
    fill=jnoun,
    line width=0.4pt,
    sharp corners,
    inner sep=2pt,
    font=\ttfamily,
    minimum width=2.8em,
    minimum height=2.8em,
    text centered,
  }
}

\newcommand{\jequiv}{\mathrel{\raisebox{0.15ex}{$\equiv$}}}
\newcommand{\jnsm}{1.1em,1.0em,0.1em,0.15em,0.15em}
\newcommand{\jnlg}{1.4em,1.0em,0.1em,0.15em,0.15em}

\newlength{\jcellwidth}
\newlength{\jcellheight}
\newlength{\jpadx}
\newlength{\jpadtop}
\newlength{\jpadbottom}

\newlength{\jnarraywidth}
\newlength{\jnarrayheight}
\newlength{\jnxleft}
\newlength{\jnytopleft}
\newlength{\jnxcenter}
\newlength{\jnycenter}
\newlength{\jntmpa}

\ExplSyntaxOn

\seq_new:N \l_jn_data_seq

\cs_new_protected:Npn \jnarrayparse #1
  {
    \seq_set_split:Nnn \l_jn_data_seq { ~ } { #1 }
  }

\cs_new:Npn \jnarrayitem #1
  {
    \seq_item:Nn \l_jn_data_seq { #1 }
  }

\ExplSyntaxOff

\def\jnsetdimensions#1,#2,#3,#4,#5\relax{%
  \setlength{\jcellwidth}{#1}%
  \setlength{\jcellheight}{#2}%
  \setlength{\jpadx}{#3}%
  \setlength{\jpadtop}{#4}%
  \setlength{\jpadbottom}{#5}%
}

\NewDocumentCommand{\jnarray}{mmmmgggg}{%
  \expandafter\jnsetdimensions#4\relax
  \begin{tikzpicture}[
    baseline={([yshift=-0.7ex]current bounding box.center)}
  ]
    \def\jncols{1}%
    \def\jncount{1}%
    \ifnum\pdfstrcmp{#1}{1x2}=0
      \def\jncols{2}%
      \def\jncount{2}%
    \fi
    \ifnum\pdfstrcmp{#1}{1x3}=0
      \def\jncols{3}%
      \def\jncount{3}%
    \fi
    \ifnum\pdfstrcmp{#1}{2x2}=0
      \def\jncols{2}%
      \def\jncount{4}%
    \fi
    \setlength{\jnarraywidth}{\jcellwidth}%
    \multiply\jnarraywidth by #3
    \setlength{\jntmpa}{\jpadx}%
    \multiply\jntmpa by 2
    \addtolength{\jnarraywidth}{\jntmpa}%
    \setlength{\jnarrayheight}{\jcellheight}%
    \multiply\jnarrayheight by #2
    \addtolength{\jnarrayheight}{\jpadtop}%
    \addtolength{\jnarrayheight}{\jpadbottom}%
    \foreach \slot in {1,2,3,4}{%
      \ifnum\slot>\jncount
      \else
        \pgfmathtruncatemacro{\row}
          {int((\slot-1)/\jncols)+1}%
        \pgfmathtruncatemacro{\col}
          {mod(\slot-1,\jncols)+1}%
        \setlength{\jnxleft}{\jnarraywidth}%
        \multiply\jnxleft by \numexpr\col-1\relax
        \setlength{\jnytopleft}{\jnarrayheight}%
        \multiply\jnytopleft by \numexpr\row-1\relax
        \jnytopleft=-\jnytopleft
        \draw[
          fill=jnoun,
          draw=black!50!white,
          line width=0.4pt
        ]
          (\jnxleft,\jnytopleft)
          rectangle
          ++(\jnarraywidth,-\jnarrayheight);
        \ifcase\slot
          \or \jnarrayparse{#5}
          \or \jnarrayparse{#6}
          \or \jnarrayparse{#7}
          \or \jnarrayparse{#8}
        \fi
        \foreach \r in {1,...,#2}{%
          \foreach \ccol in {1,...,#3}{%
            \pgfmathtruncatemacro{\idx}
              {(\r-1)*#3+\ccol}%
            \setlength{\jnxcenter}{\jnxleft}%
            \addtolength{\jnxcenter}{\jpadx}%
            \setlength{\jntmpa}{\jcellwidth}%
            \multiply\jntmpa by \numexpr\ccol-1\relax
            \addtolength{\jnxcenter}{\jntmpa}%
            \setlength{\jntmpa}{\jcellwidth}%
            \divide\jntmpa by 2
            \addtolength{\jnxcenter}{\jntmpa}%
            \setlength{\jnycenter}{\jnytopleft}%
            \addtolength{\jnycenter}{-\jpadtop}%
            \setlength{\jntmpa}{\jcellheight}%
            \multiply\jntmpa by \numexpr\r-1\relax
            \addtolength{\jnycenter}{-\jntmpa}%
            \setlength{\jntmpa}{\jcellheight}%
            \divide\jntmpa by 2
            \addtolength{\jnycenter}{-\jntmpa}%
            \edef\jnvalue{\jnarrayitem{\idx}}%
            \node[
              anchor=center,
              inner sep=0pt,
              font=\ttfamily
            ]
              at (\jnxcenter,\jnycenter)
              {\jnvalue};
          }%
        }%
      \fi
    }%
  \end{tikzpicture}%
}

\newtheorem{definition}{Definition}[section]
\newtheorem{lemma}[definition]{Lemma}
\newtheorem{theorem}[definition]{Theorem}
\newtheorem{corollary}[definition]{Corollary}

\title{Expressing NumPy Broadcasting\\ via Verb Rank in J}
\newif\ifuniqueAffiliation
\ifuniqueAffiliation 
\author{ \href{https://orcid.org/0009-0005-4803-5418}{\includegraphics[scale=0.06]{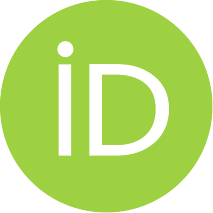}\hspace{1mm}Marcin Żołek}
}
\else
\usepackage{authblk}

\newbox{\orcid}\sbox{\orcid}{\includegraphics[scale=0.06]{orcid.pdf}} 
\author[1,2]{%
	\href{https://orcid.org/0009-0005-4803-5418}{\usebox{\orcid}\hspace{1mm}Marcin Żołek}\thanks{\texttt{mzolek@jsoftware.com}}
}
\affil[1]{Faculty of Mathematics, Informatics and Mechanics, University of Warsaw, Warsaw, Poland}
\affil[2]{Jsoftware Inc., Toronto, Canada}
\fi
\renewcommand{\undertitle}{\textcolor{darkgreen}{A Preprint
}}
\renewcommand{\shorttitle}{Expressing NumPy Broadcasting via Verb Rank in J}

\hypersetup{
pdftitle={Expressing NumPy Broadcasting via Verb Rank in J},
pdfsubject={cs.PL, cs.MS},
pdfauthor={Marcin Żołek},
pdfkeywords={array programming, broadcasting, NumPy, J programming language, multidimensional arrays}}
\begin{document}
\maketitle
\begin{abstract}
The array programming paradigm applies operations to entire arrays rather than individual elements, eliminating explicit loops and abstracting many low-level details of computation. Two influential approaches have shaped array programming: NumPy’s broadcasting and Kenneth Iverson’s array-oriented notation, with the latter first introduced in APL and later developed in the J language. Although NumPy's array programming model was influenced by Iverson's work, the two systems formalize array operations differently: NumPy primarily through array shapes and broadcasting, and J through verb rank. This paper establishes a formal correspondence between NumPy broadcasting and J rank and presents an implementation of NumPy broadcasting in J. The results provide a formal connection between two influential approaches to array-oriented computation and show how they address the common problem of applying operations to arrays with differing shapes without explicit loops.
\end{abstract}


\section{Introduction}
Array programming treats arrays as fundamental computational objects, allowing operations to be performed on entire datasets rather than individual elements. A key question in this paradigm is how a function should be applied to data when its arguments have different shapes or dimensionalities. Two influential approaches to this issue are found in J programming language and Python's NumPy library: verb rank in J \cite{hui1995} and broadcasting in NumPy \cite{vanderwalt2011}.

In J, the rank of a function determines the dimensionality of the subarrays, called cells, on which it operates. For example, a rank-2 function is applied independently to corresponding 2-dimensional cells of its arguments, with the resulting cells assembled into a single array of appropriate shape. This mechanism allows functions to operate on multidimensional data without requiring explicit iteration over the cells of their arguments \cite{burke1996}.

A separate mechanism, known as broadcasting, governs the application of operations to arrays in NumPy. Broadcasting defines rules for determining when operands with different shapes are compatible and how their dimensions are aligned. In particular, dimensions of length one can be expanded to match the corresponding dimensions of another operand \cite{broadcasting}.

Both approaches allow the computation to be expressed without explicit loops, but they achieve this through different mechanisms. This observation raises the central question of this article: how does NumPy broadcasting relate to how J applies functions to arrays? We investigate this relationship formally and show how broadcasting can be expressed in terms of verb rank.

\section{J code representation}
Throughout this article, J syntax is highlighted according to the color scheme used in the official documentation \cite{nuvoc}. Different colors distinguish parts of speech in the J programming language: \jn{nouns} (arrays) and three kinds of operators -- \jv{verbs}, \ja{adverbs}, \jc{conjunctions} -- making the code examples easier to follow. Outer rounded frames indicate when an adverb or conjunction is applied to verbs or nouns. Their color represents the part of speech of the result. For example, the expression $$\jv[r]{\jv{-}\jc{@:}\jv{+}}$$ creates a verb, so the outer frame is purple.

To keep the examples focused, the code for constructing arrays is omitted, and the arrays are represented using the format used by the Dissect add-on, which visualizes data flow in J expressions \cite{dissect}. Examples of arrays from 0-dimensional atoms to 4-dimensional hyperbricks are presented below. Italics are used to distinguish 0-dimensional arrays from 1-dimensional ones with a single element.
\[
\jn{\textit{0}}, \jn{0 1 2}, \jnarray{1x1}{2}{3}{\jnsm}{0 1 2 3 4 5}, \jnarray{1x2}{2}{2}{\jnsm}{0 1 2 3}{4 5 6 7}, \jnarray{2x2}{2}{2}
  {\jnlg}
  {0 1 2 3}
  {4 5 6 7}
  {8 9 10 11}
  {12 13 14 15}
\]
Additionally, the $\equiv$ sign is used to indicate that two expressions evaluate to the same result. For example: $$\jn{\textit{2}}\jv+\jn{\textit{3}} \jequiv \jn{\textit{5}}.$$ In particular, $\equiv$ is not part of J syntax.

\section{Methods} 
Both NumPy and J use partially different terminology for the same concepts. Some terms, such as shape, are common to both J and NumPy, while others, such as rank, were also used in the documentation of the original Numerical Python \cite{numericalpython}. For concepts shared by both, we will use a common terminology, primarily derived from J.  
\begin{definition}[Array rank]
The rank of an array is the number of its dimensions.
\end{definition}
\begin{definition}[Array shape]
The shape of an array is a rank-1 array containing the lengths of all its dimensions.
\end{definition}
\begin{definition}[$k$-cell and frame]
The $k$-cell of an array is a subarray of rank $k$ whose shape is a suffix of the array's shape. The frame is the prefix of the array's shape that precedes the shape of its $k$-cells.
\end{definition}
These definitions are illustrated in \cref{fig:defs}.
\begin{figure}[ht] 
\begin{center}
\begin{tikzpicture}
\node[draw, minimum height=6mm ] (r1) {$d_n,d_{n-1},...,d_{k+1}$};
\node[draw, minimum height=6mm, right=0mm of r1.east, anchor=west] (r2) {$d_k,...,d_2,d_1$};

\draw[decorate,decoration={brace,mirror,amplitude=4pt}]
  ($(r1.south west)+(0,-2pt)$) -- ($(r1.south east)+(0,-2pt)$)
  node[midway, below=3pt] {{\footnotesize{frame}}};

\draw[decorate,decoration={brace,mirror,amplitude=4pt}]
  ($(r2.south west)+(0,-2pt)$) -- ($(r2.south east)+(0,-2pt)$)
  node[midway, below=3pt] {{\footnotesize{shape of $k$-cells}}};
\end{tikzpicture}
\end{center}
\caption{Schematic partition of the shape $d_n,...,d_1$ of an array of rank $n$ into a frame and the shape of $k$-cells.}
\label{fig:defs}
\end{figure}
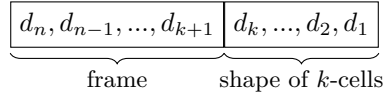

\subsection{Verb rank in J}
\label{subsec:verb-rank}
The notion of rank in array programming goes back to Kenneth Iverson's APL, a pioneering array programming language. Iverson adopted the term rank from tensor algebra, where it was used to describe the number of indices associated with a tensor \cite{iverson1962}. In APL, rank initially referred only to arrays. J retains the notion of rank for arrays and also associates ranks with functions; these are referred to as verb ranks.

A verb rank determines the cells of arrays on which the verb operates. A verb may be applied to one or two arguments. An application to one argument is called monadic, while an application to two arguments is called dyadic. In the monadic case, we say that the verb has a monadic rank $r$, meaning that it is applied to the $r$-cells of its argument, or to the entire argument if its rank is at most $r$. In the dyadic case, we say that the verb has a dyadic rank consisting of two numbers, $l$ and $r$, meaning that it operates on the $l$-cells of its left argument and the $r$-cells of its right argument, or on the entire argument when the argument's rank does not exceed the verb's rank. A verb may allow monadic application, dyadic application, or both.

For example, the verb \jv{-} (minus) has monadic rank \jn{0} and returns an array obtained by taking the arithmetic negation of each 0-cell (scalar) of the right argument. Its dyadic rank is \jn{0 0} which means that the result is constructed by subtracting each 0-cell of the right argument from the corresponding 0-cell of the left argument. The correspondence between cells when the frames of the arguments differ is defined in the following subsections. In general, the rank of a verb and its argument(s) do not determine the rank of the result. An example is the verb \jv[r]{\jv{\$}\jc{\&}\jn{0}}, which, when applied to a rank-1 argument, produces an array filled with zeros whose shape is specified by the elements of that argument.

\subsubsection{Monadic rank}
Let the verb $u$ have monadic rank $r$ and be applied to an array. The shape of the array is viewed in terms of its frame and cell shape. If the rank of the array is at most $r$, the entire array is treated as a single cell with an empty frame. Otherwise, the shape of the cells is given by the suffix of the array shape of length $r$, while the remaining prefix forms the frame. To cover both cases uniformly, we denote the rank of the cells by $r'$. The verb $u$ is then applied to each $r'$-cell. The results are collected into the frame, which means that the frame of the argument is also the frame (prefix of the shape) of the result. This is illustrated in \cref{fig:monadic}. If $u$ returns arrays of different shapes (possibly also different ranks) for different $r'$-cells of the argument, the result cells are brought to a common shape by adding leading axes as necessary and padding with the fill value, so that they can be collected into the frame. Each data type has a default fill value (we omit here the details related to the possibility of changing this value).
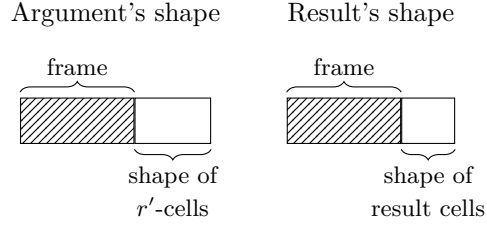
\begin{figure}
\begin{center}
\begin{tikzpicture}
\node[draw, minimum width=15mm, minimum height=6mm,
      pattern=north east lines, pattern color=black] (r1) {};
\node[draw, minimum width=10mm, minimum height=6mm, right=0mm of r1] (r2) {};

\node[draw, minimum width=15mm, minimum height=6mm,
      pattern=north east lines, pattern color=black, right=10mm of r2] (w1) {};
\node[draw, minimum width=7mm, minimum height=6mm, right=0mm of w1] (w2) {};

\coordinate (argcenter) at ($(r1.west)!0.5!(r2.east)$);
\coordinate (rescenter) at ($(w1.west)!0.5!(w2.east)$);

\node[above=32pt of argcenter] {\normalsize Argument's shape};
\node[above=32pt of rescenter] {\normalsize Result's shape};

\draw[decorate,decoration={brace,amplitude=4pt}]
  ($(r1.north west)+(0,2pt)$) -- ($(r1.north east)+(0,2pt)$)
  node[midway, above=3pt] {\footnotesize frame};

\draw[decorate,decoration={brace,mirror,amplitude=4pt}]
  ($(r2.south west)+(0,-2pt)$) -- ($(r2.south east)+(0,-2pt)$)
  node[midway, below=3pt, align=center] {\footnotesize shape of \\[-3pt] \footnotesize $r'$-cells};

\draw[decorate,decoration={brace,amplitude=4pt}]
  ($(w1.north west)+(0,2pt)$) -- ($(w1.north east)+(0,2pt)$)
  node[midway, above=3pt] {\footnotesize frame};

\draw[decorate,decoration={brace,mirror,amplitude=4pt}]
  ($(w2.south west)+(0,-2pt)$) -- ($(w2.south east)+(0,-2pt)$)
  node[midway, below=3pt, align=center] {\footnotesize shape of \\[-3pt] \footnotesize result cells};
\end{tikzpicture}
\end{center}
\caption{Schematic partitions of the argument's shape and the resulting array's shape. The same frame is a prefix of both shapes.}
\label{fig:monadic}
\end{figure}

For example, the verb \jv[r]{\jv[r]{\jv{-}\ja{/}}\jc{\ .}\jv{*}} has monadic rank \jn{2} and in the monadic case it calculates the determinant of a matrix. When the argument has shape \jn{3 2 2}, the resulting 0-cell from each matrix is collected into an array of shape \jn{3}:
\[
\jv[r]{\jv[r]{\jv{-}\ja{/}}\jc{\ .}\jv{*}} \jnarray{1x3}{2}{2}{\jnsm}{2 2 2 2}{1 2 3 4}{4 1 2 3} \jequiv \jn{0 \_2 10}.
\]
The verb \jv[r]{\jv+\ja{/}} has monadic rank equal to \jn{\_} ($\infty$). This means that, whenever this verb is applied, the rank of the argument is no greater than the rank of the verb. Hence, the entire argument is treated as a single cell. The action of this verb is equivalent to inserting \jv+ between the items of the cell, where an item is a subarray of rank one less than that of the cell. In the following example, each item has rank 1.
\[
\jv[r]{\jv+\ja{/}} \jnarray{1x1}{3}{3}{\jnsm}{0 1 2 3 4 5 6 7 8} \jequiv \jn{0 1 2}\jv+ \jn{3 4 5}\jv+\jn{6 7 8} \jequiv \jn{9 12 15}
\]
\subsubsection{Dyadic rank}
Let the verb $u$ have dyadic rank $l,r$ and be applied to two arrays. Analogously to the monadic case, the shape of each argument is viewed in terms of its frame and cell shape, with $l$ determining the rank of the left cells and $r$ determining the rank of the right cells. We denote the ranks of the left and right cells by $l'$ and $r'$, respectively. As in the monadic case, these ranks may be lower than $l$ and $r$ when the corresponding argument has lower rank. The remaining shape prefixes form the frames. We require that one of the two frames be a prefix (potentially empty or complete) of the other. We refer to the common prefix of the two frames as the common frame, i.e., the shorter of the two. The cells of the argument with the shorter frame are logically replicated so that the frames are equal, and in particular, so that the number of $l'$-cells equals the number of $r'$-cells. Next, $u$ is applied to corresponding $l'$-cells and $r'$-cells as its left and right arguments, respectively. The result cells are collected into the longer frame, as illustrated in \cref{fig:dyadic}. If the result cells have different shapes, they are brought to the same shape using the fill value, as in the monadic case.
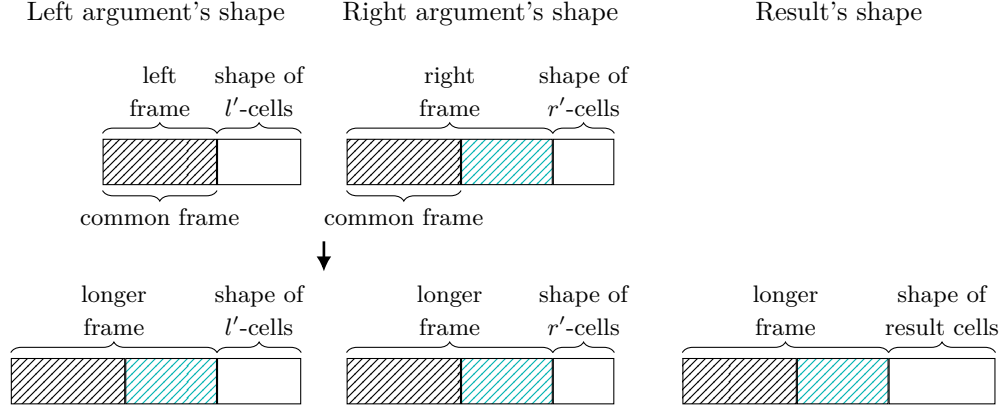
\begin{figure}
\begin{center}
\begin{tikzpicture}[
    box/.style={draw,minimum width=12mm,minimum height=6mm}
]
\coordinate (topcenter) at (0,0);
\coordinate (bottomcenter) at ($(topcenter)+(0,-29mm)$);


\node[box, left=3mm of topcenter,anchor=east,minimum width=11mm] (L1) {};
\node[box,pattern=north east lines,pattern color=black,left=0mm of L1.west,anchor=east,minimum width=15mm] (L3) {};

\node[box,pattern=north east lines,pattern color=black,right=3mm of topcenter.east,anchor=west,minimum width=15mm] (R1) {};
\node[box,pattern=north east lines,pattern color=cyan!70!black,right=0mm of R1.east,anchor=west] (R2) {};
\node[box,right=0mm of R2.east,anchor=west,minimum width=8mm] (R3) {};

\draw[decorate,decoration={brace,amplitude=4pt}]
  ($(L3.north west)+(0,2pt)$) -- ($(L3.north east)+(0,2pt)$)
  node[midway,above=3pt,align=center]
    {\footnotesize left \\[-3pt] \footnotesize frame};

\draw[decorate,decoration={brace,amplitude=4pt}]
  ($(L1.north west)+(0,2pt)$) -- ($(L1.north east)+(0,2pt)$)
  node[midway,above=3pt,align=center]
    {\footnotesize shape of \\[-3pt] \footnotesize $l'$-cells};

\draw[decorate,decoration={brace,mirror,amplitude=4pt}]
  ($(L3.south west)+(0,-2pt)$) -- ($(L3.south east)+(0,-2pt)$)
  node[midway,below=3pt]
    {\footnotesize common frame};

\draw[decorate,decoration={brace,mirror,amplitude=4pt}]
  ($(R1.south west)+(0,-2pt)$) -- ($(R1.south east)+(0,-2pt)$)
  node[midway,below=3pt]
    {\footnotesize common frame};

\draw[decorate,decoration={brace,amplitude=4pt}]
  ($(R1.north west)+(0,2pt)$) -- ($(R2.north east)+(0,2pt)$)
  node[midway,above=3pt,align=center]
    {\footnotesize right \\[-3pt] \footnotesize frame};

\draw[decorate,decoration={brace,amplitude=4pt}]
  ($(R3.north west)+(0,2pt)$) -- ($(R3.north east)+(0,2pt)$)
  node[midway,above=3pt,align=center]
    {\footnotesize shape of \\[-3pt] \footnotesize $r'$-cells};

\node[box,left=3mm of bottomcenter,anchor=east,minimum width=11mm] (bL1) {};
\node[box,pattern=north east lines,pattern color=cyan!70!black,
      left=0mm of bL1.west,anchor=east] (bL2) {};
\node[box,pattern=north east lines,pattern color=black,
      left=0mm of bL2.west,anchor=east,minimum width=15mm] (bL3) {};

\node[box,pattern=north east lines,pattern color=black,
      right=3mm of bottomcenter.east,anchor=west,minimum width=15mm] (bR1) {};
\node[box,pattern=north east lines,pattern color=cyan!70!black,
      right=0mm of bR1.east,anchor=west] (bR2) {};
\node[box,right=0mm of bR2.east,anchor=west,minimum width=8mm] (bR3) {};

\node[box,pattern=north east lines,pattern color=black,
      right=9mm of bR3.east,anchor=west,minimum width=15mm] (bW1) {};
\node[box,pattern=north east lines,pattern color=cyan!70!black,
      right=0mm of bW1.east,anchor=west] (bW2) {};
\node[box,right=0mm of bW2.east,anchor=west,minimum width=14mm] (bW3) {};

\coordinate (leftargbottomcenter) at
  ($(bL3.west)!0.5!(bL1.east)$);

\coordinate (leftargcenter) at
  (leftargbottomcenter |- topcenter);

\coordinate (rightargcenter) at
  ($(R1.west)!0.5!(R3.east)$);

\node[above=48pt of leftargcenter]
  {\normalsize Left argument's shape};

\node[above=48pt of rightargcenter]
  {\normalsize Right argument's shape};

\coordinate (resultbottomcenter) at
  ($(bW1.west)!0.5!(bW3.east)$);

\coordinate (resultcenter) at
  (resultbottomcenter |- topcenter);

\node[above=48pt of resultcenter]
  {\normalsize Result's shape};

\draw[decorate,decoration={brace,amplitude=4pt}]
  ($(bL1.north west)+(0,2pt)$) -- ($(bL1.north east)+(0,2pt)$)
  node[midway,above=3pt,align=center]
    {\footnotesize shape of \\[-3pt] \footnotesize $l'$-cells};

\draw[decorate,decoration={brace,amplitude=4pt}]
  ($(bL3.north west)+(0,2pt)$) -- ($(bL2.north east)+(0,2pt)$)
  node[midway,above=3pt,align=center]
    {\footnotesize longer \\[-3pt] \footnotesize frame};

\draw[decorate,decoration={brace,amplitude=4pt}]
  ($(bR3.north west)+(0,2pt)$) -- ($(bR3.north east)+(0,2pt)$)
  node[midway,above=3pt,align=center]
    {\footnotesize shape of \\[-3pt] \footnotesize $r'$-cells};

\draw[decorate,decoration={brace,amplitude=4pt}]
  ($(bR1.north west)+(0,2pt)$) -- ($(bR2.north east)+(0,2pt)$)
  node[midway,above=3pt,align=center]
    {\footnotesize longer \\[-3pt] \footnotesize frame};

\draw[decorate,decoration={brace,amplitude=4pt}]
  ($(bW3.north west)+(0,2pt)$) -- ($(bW3.north east)+(0,2pt)$)
  node[midway,above=3pt,align=center]
    {\footnotesize shape of \\[-3pt] \footnotesize result cells};

\draw[decorate,decoration={brace,amplitude=4pt}]
  ($(bW1.north west)+(0,2pt)$) -- ($(bW2.north east)+(0,2pt)$)
  node[midway,above=3pt,align=center]
    {\footnotesize longer \\[-3pt] \footnotesize frame};

\draw[-{Latex[length=2mm,width=2mm]},line width=1pt]
  ($(topcenter)+(0,-10.5mm)$) -- ($(bottomcenter)+(0,14.5mm)$);
\end{tikzpicture}
\end{center}
\caption{Schematic partitions of the arguments' shapes and the resulting array's shape. If one argument's frame is a prefix of the other, the longer frame is a prefix of the result's shape.}
\label{fig:dyadic}
\end{figure}

For example, the verb \jv{+} has dyadic rank \jn{0 0}, so the left and right frames form the entire shapes of the left and right arguments, respectively.
\[
\jn{\textit{2}}\jv{+}\jn{1 4 8}\jequiv\jn{3 6 10}
\]
In the above example left argument is an atom, so it has an empty frame, and the right argument has a frame \jn{3}. Hence, after agreement, the frame of the result is also \jn{3}. The result is constructed by adding the logically replicated atom from the left argument to each element of the right argument.

The following example shows agreement of the frames \jn{2 3} and \jn{2}, where the right argument is logically extended to \jnarray{1x1}{2}{3}{\jnlg}{0 0 0 10 10 10}.
\[
\jnarray{1x1}{2}{3}{\jnsm}{0 1 2 3 4 5}\jv{+}\jn{0 10} \jequiv \jnarray{1x1}{2}{3}{\jnlg}{0 1 2 13 14 15}
\]
\subsubsection{Creating verbs with custom rank}
Each verb has a rank that determines how its arguments are viewed in terms of a frame and cells on which the verb operates. J provides the rank conjunction, denoted by \jc", to explicitly specify the rank according to which the frame and cells are determined. Given a verb \jv{u} and rank value(s) \jn{n}, the expression \jv[r]{\jv{u}\jc"\jn{n}} creates a new verb of rank \jn{n} that applies \jv{u}, with its own rank, to the cells determined by \jn{n}, following the monadic and dyadic rank rules.

For example, applying \jc{"} to the verb \jv[r]{\jv+\ja{/}}, described in the subsection on monadic rank, generalizes this operation for summation along any axis of the argument. The calculation of the sum of each row, by inserting \jv+ between the items of each 1-cell, is performed by
\[
\jv[r]{\jv[r]{\jv+\ja/}\jc"\jn1}\jnarray{1x2}{2}{3}{\jnlg}{0 1 2 3 4 5}{6 7 8 9 10 11} \jequiv \jnarray{1x1}{2}{2}{\jnlg}{3 12 21 30},
\]
while the sum of each column of each 2-cell is calculated by
\[
\jv[r]{\jv[r]{\jv+\ja/}\jc"\jn2}\jnarray{1x2}{2}{3}{\jnlg}{0 1 2 3 4 5}{6 7 8 9 10 11} \jequiv \jnarray{1x1}{2}{3}{\jnlg}{3 5 7 15 17 19}.
\]
Whereas creating a rank-0 verb will result in this case in an identity function.
\[
\jv[r]{\jv[r]{\jv+\ja/}\jc"\jn0}\jnarray{1x2}{2}{3}{\jnlg}{0 1 2 3 4 5}{6 7 8 9 10 11} \jequiv \jnarray{1x2}{2}{3}{\jnlg}{0 1 2 3 4 5}{6 7 8 9 10 11}
\]
The rank conjunction also generalizes dyadic applications of verbs. The following example illustrates the application of \jv[r]{\jv{+}\jc"\jn{1 2}} to arrays of shapes \jn2 and \jn{2 2 3}, which causes an empty frame to agree with the frame \jn2. The \jv+ verb is applied to a logically replicated left argument and to each 2-cell of the right argument. Each application of \jv+ results in a further frame agreement between frames \jn2 and \jn{2 3}.
\begin{center}
\jn{0 10}\jv[r]{\jv{+}\jc"\jn{1 2}}\jnarray{1x2}{2}{3}{\jnlg}{0 1 2 3 4 5}{6 7 8 9 10 11} $\jequiv$
\jn{\jn{0 10}\jv+\jnarray{1x1}{2}{3}{\jnsm}{0 1 2 3 4 5}}\jn{\jn{0 10}\jv+\jnarray{1x1}{2}{3}{\jnlg}{6 7 8 9 10 11}} $\jequiv$ \jnarray{1x2}{2}{3}{\jnlg}{0 1 2 13 14 15}{6 7 8 19 20 21}
\end{center}
Adding each element of the left argument separately to the entire right argument is written as
\[
\jn{0 10}\jv[r]{\jv{+}\jc"\jn{0 \_}}\jnarray{1x2}{2}{3}{\jnlg}{0 1 2 3 4 5}{6 7 8 9 10 11} \jequiv \jnarray{2x2}{2}{3}{\jnlg}{0 1 2 3 4 5}{6 7 8 9 10 11}{10 11 12 13 14 15}{16 17 18 19 20 21}.
\]
The above expressions produce different results from applying the verb directly, without \jc":
\[
\jn{0 10}\jv{+}\jnarray{1x2}{2}{3}{\jnlg}{0 1 2 3 4 5}{6 7 8 9 10 11} \jequiv \jnarray{1x2}{2}{3}{\jnlg}{0 1 2 3 4 5}{16 17 18 19 20 21}.
\]
Since the conjunction \jc" produces a new verb of a given rank, it can be applied multiple times.
\[
\jn{0 10}\jv[r]{\jv[r]{\jv{+}\jc"\jn{0 \_}}\jc"\jn{1 2}}\jnarray{1x2}{2}{3}{\jnlg}{0 1 2 3 4 5}{6 7 8 9 10 11} \jequiv \jnarray{2x2}{2}{3}{\jnlg}{0 1 2 3 4 5}{10 11 12 13 14 15}{6 7 8 9 10 11}{16 17 18 19 20 21}.
\]
\subsection{Broadcasting in NumPy universal functions}
In NumPy, universal functions (ufuncs) operate on array inputs according to an associated signature that specifies how the dimensions of their arguments and result are related. For example, the dot product implemented by \texttt{numpy.vecdot} has the signature
$$(n),(n) \rightarrow (),$$
corresponding to taking the dot product of two $n$-element rank-1 arrays and producing a scalar result. For element-wise ufuncs such as \texttt{numpy.add}, corresponding to the infix operator \texttt{+}, the signature is \texttt{None}, meaning that the operation maps scalar inputs to a scalar output. The signature may contain dimension labels marked with a question mark (?) to denote optional dimensions. For example, the signature of \texttt{numpy.matmul} is $$(n?,k),(k,m?)\rightarrow(n?,m?),$$ indicating that the operation can represent matrix–matrix, matrix–vector, vector–matrix multiplication, or a dot product.

The dimensions explicitly specified by the signature define the shape of each cell. When a ufunc is applied to arrays with more dimensions than those specified by its signature, the preceding dimensions of each argument form the frame. This is analogous to J, where the signature serves a role similar to verb rank. The key difference lies in how frame agreement is handled: ufuncs use broadcasting to make differing frames equal. The following definitions formally introduce broadcasting and its associated terminology.
 
\begin{definition}[Broadcast compatibility and broadcast dimensions]
Let $A$ and $B$ be arrays to which a ufunc is applied. Let
\[
(a_m,\ldots,a_2,a_1)
\quad\text{and}\quad
(b_n,\ldots,b_2,b_1)
\]
be the frames of $A$ and $B$, respectively, as determined by the
ufunc signature. If $m \ne n$, dimensions of length $1$ are prepended
to the shorter frame so that both frames have the same length $M=\max(m,n)$. Denote the resulting frames by
\[
(a_M,\ldots,a_2,a_1)
\quad\text{and}\quad
(b_M,\ldots,b_2,b_1).
\]
The arrays $A$ and $B$ are broadcast compatible with respect to the ufunc if and only if, for every
$i\in\{1,\ldots,M\}$,
\[
a_i=b_i
\quad\text{or}\quad
a_i=1
\quad\text{or}\quad
b_i=1.
\]
For $i\in\{1,\ldots,M\}$, $a_i$ is called a broadcast dimension in $A$ with respect to $B$ if
\[
a_i=1 \quad\text{and}\quad b_i\ne 1.
\]
Similarly, $b_i$ is called a broadcast dimension in $B$ with respect to $A$ if
\[
b_i=1 \quad\text{and}\quad a_i\ne 1.
\]
\end{definition}
The above definitions are illustrated by the example shown in \cref{fig:coreloop}.
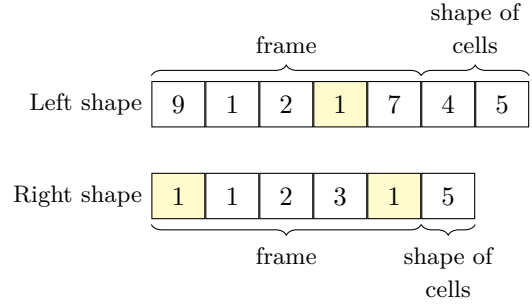
\begin{figure}[ht]
\centering
\begin{tikzpicture}[
   box/.style={
       draw,
       minimum width=7mm,
       minimum height=6mm,
       inner sep=1pt,
       align=center
   },
   onebox/.style={
       box,
       fill=yellow!25
   }
]

\usetikzlibrary{decorations.pathreplacing}

\def\vgap{3mm}

\node[box,    anchor=south west] (T1) at (-17.5mm,\vgap) {9};
\node[box,    anchor=west]       (T2) at (T1.east)       {1};
\node[box,    anchor=west]       (T3) at (T2.east)       {2};
\node[onebox, anchor=west]       (T4) at (T3.east)       {1};
\node[box,    anchor=west]       (T5) at (T4.east)       {7};
\node[box,    anchor=west]       (T6) at (T5.east)       {4};
\node[box,    anchor=west]       (T7) at (T6.east)       {5};

\node[onebox, anchor=north west] (B1) at (-17.5mm,-\vgap) {1};
\node[box, anchor=west]       (B2) at (B1.east)           {1};
\node[box,    anchor=west]       (B3) at (B2.east)        {2};
\node[box,    anchor=west]       (B4) at (B3.east)        {3};
\node[onebox, anchor=west]       (B5) at (B4.east)        {1};
\node[box,    anchor=west]       (B6) at (B5.east)        {5};

\draw[decorate,decoration={brace,amplitude=4pt}]
  ($(T6.north west)+(0,2pt)$) -- ($(T7.north east)+(0,2pt)$)
  node[midway,above=5pt,align=center]{\footnotesize shape of\\[-3pt]\footnotesize cells};

\draw[decorate,decoration={brace,amplitude=4pt}]
  ($(T1.north west)+(0,2pt)$) -- ($(T5.north east)+(0,2pt)$)
  node[midway,above=5pt,align=center]{\footnotesize frame};

\draw[decorate,decoration={brace,mirror,amplitude=4pt}]
  ($(B6.south west)+(0,-2pt)$) -- ($(B6.south east)+(0,-2pt)$)
  node[midway,below=5pt,align=center] {\footnotesize shape of\\[-3pt]\footnotesize cells};

\draw[decorate,decoration={brace,mirror,amplitude=4pt}]
  ($(B1.south west)+(0,-2pt)$) -- ($(B5.south east)+(0,-2pt)$)
  node[midway,below=5pt,align=center]{\footnotesize frame};

\node[anchor=east] at (T1.west) {$\text{\footnotesize Left shape}$};
\node[anchor=east] at (B1.west) {$\text{\footnotesize Right shape}$};
\end{tikzpicture}
\caption{Schematic partition of the shapes of example broadcast compatible arrays with respect to the ufunc with signature $(m,n),(n)\rightarrow(m)$. Broadcast dimensions are highlighted in yellow. Note that not all dimensions of length 1 are broadcast dimensions.}
\label{fig:coreloop}
\end{figure}
\begin{definition}[Broadcasting]
Let $A$ and $B$ be arrays that are broadcast compatible with respect to a ufunc. Each broadcast dimension is expanded to match the length of the corresponding non-broadcast dimension in the other array's frame. Values along the expanded dimensions are logically replicated. Consequently, the two arrays have the same frame.
\end{definition}

After broadcasting, the ufunc is executed independently on corresponding pairs of cells. The signature guarantees that all result cells have the same shape. 

Listing~\ref{lst:numpybroadcasting} presents examples of broadcasting in NumPy that correspond to some of the examples in J from the \cref{subsec:verb-rank}.
\begin{lstlisting}[mathescape,caption={Examples of broadcasting in the Python console.},captionpos=t,label={lst:numpybroadcasting}] 
>>> import numpy as np
>>> 2 + np.array([1,4,8])
array([ 3,  6, 10])
>>> np.arange(6).reshape(2,3) + np.array([0,10]).reshape(2,1)
array([[ 0,  1,  2],
       [13, 14, 15]])
>>> np.array([0,10]).reshape(2,1) + np.arange(12).reshape(2,2,3)
array([[[ 0,  1,  2],
        [13, 14, 15]],

       [[ 6,  7,  8],
        [19, 20, 21]]])
\end{lstlisting}
\section{Results}
NumPy’s broadcasting can be expressed in J as an adverb \ja{Broadcastly} such that \jn{x}\jv[r]{\jv{u}\ja{Broadcastly}}\jn{y} returns the same result as \texttt{u(x, y)} in NumPy, where
\begin{itemize}
\item \jn{x} and \jn{y} are J arrays having the same shapes and values as \texttt{x} and \texttt{y} in NumPy, 
\item \jv{u} is a verb with dyadic rank $l, r$, where $l$ and $r$ are the numbers of specified dimensions of the first and second arguments, respectively, in the signature of the ufunc \texttt{u}. For example, when the signature is \texttt{None}, the verb must have dyadic rank $\jn{0 0}$, whereas for the signature $(n),(n,m?)\rightarrow(m?)$, the corresponding rank is $\jn{1 2}$. For any left and right arguments of ranks $l' \leq l$ and $r' \leq r$, respectively, for which the ufunc \texttt{u} is defined, \jv{u} returns the same result as \texttt{u}.
\end{itemize}
 
The requirement that \jv{u} return the same result as the ufunc \texttt{u} whenever the ufunc is defined for the given cells follows from the minimum argument ranks specified by the ufunc signature. In a ufunc signature, the minimum rank of each argument is equal to the number of dimension labels not marked with question mark ($?$) in the corresponding signature part. J's dyadic rank does not impose such a minimum. As a result, a ufunc may raise an error when an operand does not have enough dimensions to satisfy the signature, whereas the corresponding J verb still produces a result. This difference can occur only for ufuncs that are not element-wise. \Cref{tab:operators} illustrates the correspondence between selected NumPy ufuncs and J verbs under these conditions.
\begin{table}[ht]
 \caption{Examples of the correspondence between ufuncs in NumPy and verbs in J to which the \ja{Broadcastly} is applied. For consistency, we represent a \texttt{None} signature as $(),()\rightarrow()$.}
\centering
\begin{tabular}{lc@{\hspace{2em}}cc}
\multicolumn{2}{c}{NumPy} & \multicolumn{2}{c}{J} \\
ufunc & signature & verb & dyadic rank \\
\hline
\texttt{multiply} & $(),()\rightarrow()$ & \jv* & \jn{0 0} \\
\texttt{vecdot} & $(n),(n)\rightarrow()$ & \jv{+}\ja/\jc{@:}\jv*\jc"\jn1 & \jn{1 1} \\
\texttt{matmul} & $(n?,k),(k,m?)\rightarrow(n?,m?)$ & \jv+\ja/\jc{\ .}\jv*\jc"\jn2 & \jn{2 2} \\
\end{tabular}
\label{tab:operators}
\end{table}
\subsection{Algorithm}
In NumPy, the broadcast dimensions present in array shapes determine how an operator applied to those arrays behaves. In J, however, there is no notion of a broadcast dimension, since dimensions of length 1 have no special semantic significance and do not determine the behavior of operators. In J, the corresponding behavior of an operator (verb) is determined by its rank. This means that \jv[r]{\jv{u}\ja{Broadcastly}} will determine, based on the shapes of its arguments, which dimensions correspond to NumPy's broadcast dimensions. Based on these dimensions, it will construct a verb of the appropriate rank and apply it to the arguments with the broadcast dimensions removed. The construction of the verb will consist in applying a sequence of rank conjunctions to \jv{u}.
  
We first consider a special case when ufuncs are element-wise, which provides the basis for the general case presented subsequently.

\subsubsection{Construction of a verb corresponding to an element-wise universal function}
\label{subsec:construction} 
Consider an element-wise ufunc in NumPy applied to two broadcast compatible arrays. In this case, the frames can be identified with the array shapes. The arrays may be assumed without loss of generality to have the same rank $n$, as dimensions of length~1 can be prepended to the shape of the lower-rank array.

We introduce an indexing scheme for the dimensions specified by array shapes, which will be used in the construction of the verb. In each array, the non-broadcast dimensions are indexed by consecutive natural numbers from right to left, starting at $1$, with the broadcast dimensions omitted from the numbering. We denote the non-broadcast dimensions by $\delta_i$ and $\widetilde{\delta}_j$ for the left and right arrays, respectively, where $i$ and $j$ are the given indices. Broadcast dimensions are denoted by~$\beta_i$ and~$\widetilde{\beta}_j$, respectively. Each broadcast dimension of the left array is assigned the largest index~$i$ such that $\delta_i$ occurs to its right, and each broadcast dimension of the right array is assigned the largest index~$j$ such that $\widetilde{\delta}_j$ occurs to its right. In either case, the value assigned is~$0$ if no such index exists. An example of this indexing scheme is shown in \cref{fig:indeksowanie}.

\begin{figure}[ht]
\centering
\begin{tikzpicture}[
   box/.style={
       draw,
       minimum width=7mm,
       minimum height=6mm,
       inner sep=1pt,
       align=center
   },
   onebox/.style={
       box,
       fill=yellow!25
   }
]
\tikzset{connectdotted/.style={densely dotted, line cap=round}}
\def\vgap{3mm}
\node[box,    anchor=south west] (T1) at (-17.5mm,\vgap) {$\delta_4$};
\node[box,    anchor=west]       (T2) at (T1.east)       {$\delta_3$};
\node[onebox, anchor=west]       (T3) at (T2.east)       {$\beta_2$};
\node[box,    anchor=west]       (T4) at (T3.east)       {$\delta_2$};
\node[box,    anchor=west]       (T5) at (T4.east)       {$\delta_1$};
 
\node[onebox, anchor=north west] (B1) at (-17.5mm,-\vgap) {$\widetilde{\beta}_2$};
\node[onebox, anchor=west]       (B2) at (B1.east)        {$\widetilde{\beta}_2$};
\node[box,    anchor=west]       (B3) at (B2.east)        {$\widetilde{\delta}_2$};
\node[box,    anchor=west]       (B4) at (B3.east)        {$\widetilde{\delta}_1$};
\node[onebox, anchor=west]       (B5) at (B4.east)        {$\widetilde{\beta}_0$};

\draw[connectdotted] (T1.south) -- (B1.north);
\draw[connectdotted] (T2.south) -- (B2.north);
\draw[connectdotted] (T3.south) -- (B3.north);
\draw[connectdotted] (T5.south) -- (B5.north);

\node[anchor=east] at (T1.west) {$\text{\footnotesize Left shape}$};
\node[anchor=east] at (B1.west) {$\text{\footnotesize Right shape}$};
\end{tikzpicture}
\caption{Indexing dimensions in example array shapes when arrays are broadcast together. Broadcast dimensions are highlighted in yellow. Dotted lines indicate the correspondence between the broadcast dimensions of one array and the non-broadcast dimensions of the other. For the arrays to be broadcast compatible, we must have $\delta_2 = \widetilde{\delta}_1$. Note that $\delta_2 = \widetilde{\delta}_1 = 1$ is possible.}
\label{fig:indeksowanie}
\end{figure}
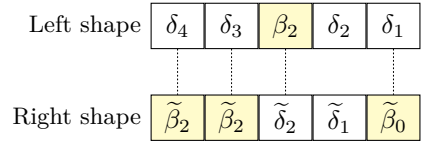
The construction of a verb of the appropriate rank for a given \jv{u} and two broadcast compatible arrays of rank $n$ proceeds by simultaneously traversing their indexed shapes from right to left. Thus, at step $1\leq k\leq n$, the $k$-th dimensions from the right in the shapes of both arrays are considered. Exactly one of the following cases applies:
\begin{itemize}
    \item $\delta_i$, $\widetilde{\delta}_j$ -- the dimension is non-broadcast in both arrays. Since the arrays are broadcast compatible, we have $\delta_i = \widetilde{\delta}_j$. In this case, we retain the verb constructed so far.
    
    \item $\delta_i$, $\widetilde{\beta}_j$ -- the right array has a broadcast dimension. We construct a new verb by applying the rank conjunction \jc{"} to the current verb and an array representing the dyadic rank $(i-1, j)$.
    
    \item $\beta_i$, $\widetilde{\delta}_j$ -- the left array has a broadcast dimension. We construct a new verb by applying the rank conjunction \jc{"} to the current verb and an array representing the dyadic rank $(i, j-1)$.
\end{itemize}
The construction is complete after all $n$ pairs of corresponding dimensions have been processed.

Following the algorithm described above, for the example in \cref{fig:indeksowanie}, we obtain the verb
\begin{equation}
\jv[r]{\jv[r]{\jv[r]{\jv[r]{\jv{u}\jc{"}\jn{0 0}}\jc{"}\jn{2 1}}\jc{"}\jn{2 2}}\jc{"}\jn{3 2}}.
\label{eq:unoptimized}
\end{equation}
We apply the verb constructed in this way to the arrays without broadcast dimensions, i.e., to the left array with shape $\delta_4\delta_3\delta_2\delta_1$ and the right array with shape $\widetilde{\delta}_2\widetilde{\delta}_1$. The following theorem establishes that applying the verb constructed by the algorithm to arrays with the broadcast dimensions removed produces a result with the same shape as that returned by the element-wise ufunc when applied to arrays with their original shapes.
\begin{theorem}\label{thm:elementwise}
Let $n$ be the rank of the left and right arrays. The arrays are assumed to be broadcast compatible with respect to an element-wise ufunc. The following invariant holds for every $0\leq k\leq n$, where $k=0$ denotes the initial state before the first step. Let $s_k$ and $\widetilde{s}_k$ denote the suffixes of length $k$ of the shapes of the left and right arrays, respectively.  For $k\geq1$, let $i$ and $j$ be the indices assigned to the $k$-th dimensions from the right in the shapes of the left and right arrays, respectively; for $k=0$, set $i=j=0$. Let $\delta_i\delta_{i-1}\ldots\delta_1$ and $\widetilde{\delta}_j\widetilde{\delta}_{j-1}\ldots\widetilde{\delta}_1$ denote the shapes obtained from $s_k$ and $\widetilde{s}_k$, respectively, after removing the broadcast dimensions. The invariant states that the following operations produce $k$-dimensional arrays of the same shape:
\begin{itemize}
    \item an element-wise ufunc applied to arrays with shapes $s_k$ and $\widetilde{s}_k$
    \item application of the currently constructed verb to arrays with shapes $\delta_i\delta_{i-1}\ldots \delta_1$ and $\widetilde{\delta}_j\widetilde{\delta}_{j-1}\ldots \widetilde{\delta}_1$.
\end{itemize}
\end{theorem}
\begin{proof}
We prove the invariant by induction on the number of steps performed.
\newline\textbf{Base case ($k = 0$).}
The shapes $s_k$ and $\widetilde{s}_k$ are empty, so the element-wise ufunc returns a scalar. Before the first iteration, the current verb is \jv{u}. Since $i=j=0$, \jv{u} is applied to arrays with empty shapes (rank 0). By the assumption that \jv{u} produces the same result as the element-wise ufunc when applied to 0-cells, the result is a rank-0 array.
\newline\textbf{Inductive case ($1\leq k\leq n$).}
Let $v$ denote the current verb, and let $s_{result}$ denote the shape of the result from the previous iteration. We consider the three cases from the algorithm:
\begin{itemize}
    \item $\delta_i$, $\widetilde{\delta}_j$ -- since the array shapes are assumed to be broadcast compatible, we have $\delta_i=\widetilde{\delta}_j$. Thus, the shape of the result of the NumPy element-wise ufunc on arrays with shapes $s_k$ and $\widetilde{s}_k$ is $\delta_is_{result}$. By the induction hypothesis, applying $v$ to arrays with shapes $\delta_{i-1}\delta_{i-2}\ldots \delta_1$ and $\widetilde{\delta}_{j-1}\widetilde{\delta}_{j-2}\ldots \widetilde{\delta}_1$ produces an array with shape $s_{result}$. We want to show that applying $v$ to arrays with shapes $\delta_i\delta_{i-1}\ldots \delta_1$ and $\widetilde{\delta}_j\widetilde{\delta}_{j-1}\ldots\widetilde{\delta}_1$ produces an array with shape $\delta_is_{result}$. The verb $v$ has dyadic rank $(l,r)$ with $l<i$ and $r<j$, so $\delta_i$ and $\widetilde{\delta_j}$ are the leftmost elements of the left and right frames, respectively. Since $\delta_i=\widetilde{\delta}_j$, $\delta_i$ is the leftmost element of the common frame, and the common frame is always a prefix of the result shape.
    
    \item $\delta_i$, $\widetilde{\beta}_j$ -- the right array has a broadcast dimension, so the shape of the result of the element-wise ufunc on arrays with shapes $s_k$ and $\widetilde{s}_k$ is $\delta_is_{result}$. By the induction hypothesis, applying $v$ to arrays with shapes $\delta_{i-1}\delta_{i-2}\ldots \delta_1$ and $\widetilde{\delta}_{j}\widetilde{\delta}_{j-1}\ldots\widetilde{\delta}_1$ produces an array with shape $s_{result}$. Let $v'$ be the new verb obtained by applying the rank conjunction to $v$ and the array representing the dyadic rank $(i-1,j)$. We want to show that applying $v'$ to arrays with shapes $\delta_i\delta_{i-1}\ldots \delta_1$ and $\widetilde{\delta}_j\widetilde{\delta}_{j-1}\ldots\widetilde{\delta}_1$ produces an array with shape $\delta_is_{result}$. The verb $v'$ determines the frame of the left array relative to $(i-1)$-cells, so the left frame is $\delta_i$. It determines the frame of the right array relative to $j$-cells, so the right frame is empty. As a result of frame agreement, $\delta_i$ becomes the frame of the result, and $v$ is applied to the corresponding $(i-1)$-cells and $j$-cells of the left and right arrays, respectively. By the induction hypothesis, the shape of the resulting cells is $s_{result}$.
    
    \item $\beta_i$, $\widetilde{\delta}_j$ -- the case is analogous to the previous one.
\end{itemize}
By induction, the invariant holds before the first step and after each subsequent step of the algorithm.
\end{proof} 
\subsubsection{Minimizing the number of rank conjunctions}
Let's examine when an application of the rank conjunction can be omitted from the algorithm above in order to reduce the number of compositions. Let a block of broadcast dimensions be a contiguous subsequence of the shape of an array, where each dimension in the subsequence is a broadcast dimension. A block that cannot be extended is called maximal. Suppose that a maximal block occurs in one of the shapes. By symmetry, assume that it occurs in the shape of the right argument as $\widetilde{\beta}_k\widetilde{\beta}_k...\widetilde{\beta}_k$ for some $k$, while the corresponding positions in the shape of the left argument contain the dimensions $\delta_j\delta_{j-1}...\delta_{i+1}\delta_i$. According to the algorithm, this block corresponds to a sequence of compositions with rank conjunctions for the dyadic ranks $(i-1,k)$, $(i,k)$, $\ldots$, $(j-2,k)$, $(j-1,k)$, respectively. The repeated frame agreements resulting from these ranks are equivalent to determining the frame of the left array relative to $(i-1)$-cells and the frame of the right array relative to $k$-cells. Thus, for a maximal block, a single application of the rank conjunction suffices, with dyadic rank $(i-1,k)$. This means that, in general, it is sufficient to apply only those rank conjunctions that, in the algorithm, correspond to the rightmost broadcast dimensions of each maximal block. The number of compositions is then equal to the number of maximal blocks of broadcast dimensions in both arrays. Additionally, if the innermost composition (the first composition computed by the algorithm) with a rank conjunction concerns the ranks $\jn{0 0}$, it can also be omitted, since verbs in J corresponding to element-wise ufuncs in NumPy have rank \jn{0 0}. Taking these two optimizations into account, the verb~\eqref{eq:unoptimized} can be simplified to the form~\eqref{eq:optimized}.
\begin{equation}
\jv[r]{\jv[r]{\jv{u}\jc{"}\jn{2 1}}\jc{"}\jn{2 2}}
\label{eq:optimized}
\end{equation}
\subsubsection{Generalizing to all universal functions}\label{subsec:generalization} 
The result of theorem~\ref{thm:elementwise} can be extended from element-wise ufuncs to ufuncs with general signatures. We show that the sequence of ranks constructed by the algorithm can be adapted to higher-rank cells when the algorithm is applied to frames.
\begin{lemma}\label{lem:increasedranks}
Let $v$ be a verb of dyadic rank $(0, 0)$, and let $V$ be a verb of dyadic rank $(l,r)$. Let $l'$ and $r'$ satisfy $0\leq l'\leq l$ and $0\leq r'\leq r$. Let $\delta_i$, $\widetilde{\delta}_i$, $\gamma_i$, and $\widetilde{\gamma}_i$ be arbitrary dimension lengths. Assume that every valid application of $v$ to a pair of 0-cells as the left and right arguments, respectively, produces a 0-cell, and that every valid application of $V$ to an $l'$-cell of shape $\gamma_{l'}\ldots\gamma_1$ and an $r'$-cell of shape $\widetilde{\gamma}_{r'}\ldots\widetilde{\gamma}_1$ as the left and right arguments, respectively, produces an array of shape $\Gamma$. Let $(l_i)_{i\geq 1}$ and $(r_i)_{i\geq 1}$ be sequences satisfying
\[
l'\leq l_1\leq l_2\leq\cdots
\qquad\text{and}\qquad
r'\leq r_1\leq r_2\leq\cdots.
\]

Then, for every $k\geq 0$, there exists a frame $\Delta_k$ such that the following implication holds: if applying
$$ 
v\texttt{"}(l_1,r_1)\texttt{"}(l_2, r_2)\texttt{"}\ldots\texttt{"}(l_k,r_k)
$$
to arrays of shapes
$$
\delta_{l_{k+1}}\ldots\delta_1
\quad\text{and}\quad
\widetilde{\delta}_{r_{k+1}}\ldots\widetilde{\delta}_1
$$
as the left and right arguments, respectively, produces an array of shape $\Delta_k$, then applying
$$
V\texttt{"}(l',r')\texttt{"}(l_1+l',r_1+r')\texttt{"}\ldots\texttt{"}(l_k+l',r_k+r')
$$
to arrays of shapes
$$
\delta_{l_{k+1}}\ldots\delta_1\gamma_{l'}\ldots\gamma_1 \quad\text{and}\quad
\widetilde{\delta}_{r_{k+1}}\ldots\widetilde{\delta}_1\widetilde{\gamma}_{r'}\ldots\widetilde{\gamma}_1
$$
as the left and right arguments, respectively, produces an array of shape \(\Delta_k\Gamma\).
\end{lemma}
\begin{proof}
We prove the lemma by induction on the $k$.
\newline\textbf{Base case ($k = 0$).}
Suppose that applying $v$ to arrays of shapes
$$
\delta_{l_1}\ldots\delta_1
\quad\text{and}\quad
\widetilde{\delta}_{r_1}\ldots\widetilde{\delta}_1
$$
produces an array of shape $\Delta_0$. Since $v$ has dyadic rank $(0, 0)$ and produces a $0$-cell when applied to two $0$-cells, $\Delta_0$ is the longer of the two shapes, as frame agreement selects the longer frame.
Applying $V\texttt{"}(l',r')$ to arrays of shapes
$$\delta_{l_1}\ldots\delta_1\gamma_{l'}\ldots\gamma_1
\quad\text{and}\quad 
\widetilde{\delta}_{r_1}\ldots\widetilde{\delta}_1
\widetilde{\gamma}_{r'}\ldots\widetilde{\gamma}_1$$
first agrees the frames $\delta_{l_1}\ldots\delta_1$ and $\widetilde{\delta}_{r_1}\ldots\widetilde{\delta}_1$. Thus, the result has frame $\Delta_0$, and the corresponding $l'$-cells and $r'$-cells are then passed to $V$. By assumption, applying $V$ to such cells produces an array of shape $\Gamma$. Hence, the resulting array has shape $\Delta_0\Gamma$.
\newline\textbf{Inductive case ($k\geq 1$).}
Assume that applying
$$
v\texttt{"}(l_1,r_1)\texttt{"}\ldots\texttt{"}(l_k,r_k)
$$
to arrays of shapes
$$ 
\delta_{l_{k+1}}\ldots\delta_1
\quad\text{and}\quad
\widetilde{\delta}_{r_{k+1}}\ldots\widetilde{\delta}_1
$$
produces an array of shape $\Delta_k$. Their frames are
$$
\delta_{l_{k+1}}\ldots\delta_{l_k+1}
\quad\text{and}\quad
\widetilde{\delta}_{r_{k+1}}\ldots\widetilde{\delta}_{r_k+1}.
$$
Frame agreement selects the longer frame. By symmetry, it suffices to consider the case in which $\delta_{l_{k+1}}\ldots\delta_{l_k+1}$ is the longer frame. Then
$$
\Delta_k=\delta_{l_{k+1}}\ldots\delta_{l_k+1}\Delta_{k-1},
$$
where $\Delta_{k-1}$ is the shape of the result of applying
$$
v\texttt{"}(l_1,r_1)\texttt{"}\ldots\texttt{"}(l_{k-1},r_{k-1})
$$
to cells of shapes
$$
\delta_{l_k}\ldots\delta_1
\quad\text{and}\quad
\widetilde{\delta}_{r_k}\ldots\widetilde{\delta}_1.
$$
 
Thus, the antecedent of the induction hypothesis for $k-1$ is satisfied. Hence, applying
$$
V\texttt{"}(l',r')\texttt{"}(l_1+l',r_1+r')\texttt{"}\ldots\texttt{"}(l_{k-1}+l',r_{k-1}+r')
$$
to arrays of shapes
$$
\delta_{l_k}\ldots\delta_1\gamma_{l'}\ldots\gamma_1
\quad\text{and}\quad
\widetilde{\delta}_{r_k}\ldots\widetilde{\delta}_1
\widetilde{\gamma}_{r'}\ldots\widetilde{\gamma}_1
$$
produces an array of shape $\Delta_{k-1}\Gamma$.

Now consider applying
$$
V\texttt{"}(l',r')\texttt{"}(l_1+l',r_1+r')\texttt{"}\ldots\texttt{"}(l_k+l',r_k+r')
$$
to arrays of shapes
$$
\delta_{l_{k+1}}\ldots\delta_1\gamma_{l'}\ldots\gamma_1 \quad\text{and}\quad
\widetilde{\delta}_{r_{k+1}}\ldots\widetilde{\delta}_1\widetilde{\gamma}_{r'}\ldots\widetilde{\gamma}_1.
$$
This first agrees the frames
$$
\delta_{l_{k+1}}\ldots\delta_{l_k+1}
\quad\text{and}\quad
\widetilde{\delta}_{r_{k+1}}\ldots\widetilde{\delta}_{r_k+1},
$$
so the longer frame
$\delta_{l_{k+1}}\ldots\delta_{l_k+1}$ is selected. It then applies
$$
V\texttt{"}(l',r')\texttt{"}(l_1+l',r_1+r')\texttt{"}\ldots\texttt{"}(l_{k-1}+l',r_{k-1}+r')
$$
to the corresponding $(l_k+l')$-cells and $(r_k+r')$-cells. Each such pair of cells has shape
$$\delta_{l_k}\ldots\delta_1\gamma_{l'}\ldots\gamma_1\quad\text{and}\quad \widetilde{\delta}_{r_k}\ldots\widetilde{\delta}_1 \widetilde{\gamma}_{r'}\ldots\widetilde{\gamma}_1,$$
respectively, so the induction hypothesis applies and gives an array of shape $\Delta_{k-1}\Gamma$. Hence, the resulting shape is
$$
\delta_{l_{k+1}}\ldots\delta_{l_k+1}\Delta_{k-1}\Gamma
=\Delta_k\Gamma.
$$
\end{proof}
\begin{corollary}
Let $s$ and $\widetilde{s}$ be the frames of two broadcast compatible arrays. Let $\delta_m\ldots\delta_1$ and $\widetilde{\delta}_n\ldots\widetilde{\delta}_1$ be the frames obtained from $s$ and $\widetilde{s}$, respectively, by removing their broadcast dimensions. Let $v$ be a verb of dyadic rank $(0,0)$ that produces a $0$-cell when applied to a pair of $0$-cells. The construction from \cref{subsec:construction} applied to $v$ and arrays of shapes $s$ and $\widetilde{s}$ produces the verb
\[
v\texttt{"}(l_1,r_1)\texttt{"}\ldots\texttt{"}(l_N,r_N).
\]
Let $\Delta$ be the shape of the result obtained by applying this verb to arrays of shapes $\delta_m\ldots\delta_1$ and $\widetilde{\delta}_n\ldots\widetilde{\delta}_1$.

Let $U$ be a ufunc with a signature containing $l$ and $r$ dimension labels for its left and right arguments, respectively. Let $V$ be a verb of dyadic rank $(l,r)$ that produces the same result as $U$ on arguments of ranks at most $l$ and $r$, respectively, whenever $U$ is valid.

Suppose that $U$ is applied validly to arrays of shapes
\[
s\gamma_{l'}\ldots\gamma_1
\quad\text{and}\quad
\widetilde{s}\widetilde{\gamma}_{r'}\ldots\widetilde{\gamma}_1,
\]
where $s$ and $\widetilde{s}$ are the frames, which broadcast to $\Delta$, and
$l'\leq l$ and $r'\leq r$ are the ranks of the left and right cells, respectively. The cells have shapes $\gamma_{l'}\ldots\gamma_1$ and
$\widetilde{\gamma}_{r'}\ldots\widetilde{\gamma}_1$. As specified by its signature, $U$ maps these cells to a cell of shape $\Gamma$. Hence, $U$ produces an array of shape $\Delta\Gamma$.

Then applying the verb
\[
V\texttt{"}(l',r')\texttt{"}(l_1+l',r_1+r')\texttt{"}\ldots\texttt{"}(l_N+l',r_N+r')
\]
to arrays of shapes
\[
\delta_m\ldots\delta_1\gamma_{l'}\ldots\gamma_1
\quad\text{and}\quad
\widetilde{\delta}_n\ldots\widetilde{\delta}_1\widetilde{\gamma}_{r'}\ldots\widetilde{\gamma}_1
\] 
produces an array of the same shape $\Delta\Gamma$. 
\end{corollary}
\begin{proof}
By theorem~\ref{thm:elementwise}, the constructed verb $v\texttt{"}(l_1,r_1)\texttt{"}\ldots\texttt{"}(l_N,r_N)$, when applied to arrays of shapes $\delta_m\ldots\delta_1$ and $\widetilde{\delta}_n\ldots\widetilde{\delta}_1$, produces an array of shape $\Delta$. Thus, in the notation of lemma~\ref{lem:increasedranks}, we may take $\Delta_N=\Delta$, $l_{N+1}=m$, $r_{N+1}=n$. Applying lemma~\ref{lem:increasedranks} with $k=N$ then gives the stated result.
\end{proof}

\subsection{Implementation}
The implementation of broadcasting using the generalized algorithm is presented in listing~\ref{lst:broadcastly}. The part of speech defined inside $\jt{\{\{}$, $\jt{\}\}}$ is recognized by the interpreter based on the special names used within it: $\jn{x}$, $\jn{y}$, $\jv{u}$, $\jv{v}$, $\jn{m}$ and $\jn{n}$. In particular, the presence of $\jv{u}$, $\jn{x}$, $\jn{y}$ in the definition of $\ja{Broadcastly}$ indicates that it is an adverb:
\begin{itemize}
    \item $\jv{u}$ denotes the verb to which the adverb will be applied.
    
    \item $\jn{x}$ and $\jn{y}$ denote the left and right arguments of the verb $\jv[r]{\jv{u}\ja{Broadcastly}}$, respectively.
\end{itemize}
In line 2, the dyadic rank of the \jv{u} is inspected, and the left and right ranks are assigned to the names \jn{l} and \jn{r}, respectively. Lines 3–4 define the verbs \jv{cs} and \jv{fr}. Given the value of rank as the left argument and an array as the right argument, \jv{cs} returns the shape of the cells, while \jv{fr} returns the frame. In line 5, the shapes of the left and right cells are computed using \jv{cs} and assigned to \jn{csx} and \jn{csy}, respectively. In line 6, the frames of the arguments across which broadcasting will be applied are computed using \jv{fr}. The frames are then aligned by prepending dimensions of length 1 to the frame of the argument with the lower rank, ensuring that both frames have the same length. The resulting frames are stored in \jn{fxy}, a two-row matrix whose first row corresponds to the frame of the left argument and whose second row corresponds to that of the right argument. The broadcast compatibility of the two frames is then checked. In line 8, the boolean matrix \jn{mxy} is constructed to identify the broadcast dimensions. It has the same shape as \jn{fxy}. This matrix is used in line 9 to compute \jn{ranks}, a two-column matrix containing pairs of left and right dyadic ranks, treating the frames as the entire array shapes. In line 10, these ranks are updated to account for the full array shapes, as described in \cref{subsec:generalization}. The verb's representation is constructed in line 11. Line 12 computes the new frames by removing the broadcast dimensions. Finally, these frames are used in line 13 to obtain the result by applying the constructed verb to arrays having the corresponding frames.
\begin{center}
\begin{minipage}{0.95\linewidth}
\begin{lstlisting}[mathescape,caption={Implementation of broadcasting in J.},captionpos=t,numbers=left,numberstyle=\tiny,label={lst:broadcastly}]
$\ja{Broadcastly}\jo{=:}\jt{\{\{}$
  $\jn{'l r'}\jo{=.}\jv{\}.}\ \jv{u}\ja{b.}\ \jn0$
  $\jv{cs}\jo{=.}\bo\bo\jv-\jc{@}\jv{<.}\ \jv{\#}\bc\ \jv{\{.}\ \jv{]}\bc\ \jv{\$}$
  $\jv{fr}\jo{=.}\jv-\jc{@}\jv{[}\ \jv{\}.}\ \jv{\$}\jc{@}\jv{]}$
  $\jn{csx}\jo{=.}\jn{l}\jv{cs}\jn{x}\ \jv{[}\ \jn{csy}\jo{=.}\jn{r}\jv{cs}\jn{y}$
  $\jn{fxy}\jo{=.}\bo\jn{l}\jv{fr}\jn{x}\bc\ \jv{,:}\jc{!.}\jn{1}\jc{\&}\jv{|.}\ \jn{r}\jv{fr}\jn{y}$
  $\jt{assert.}\bo\jv{+.}\ja{/}\jc{@:}\bo\jv{=}\jc{\&}\jn{1}\bc\ \jv{*.}\ja{/}\jc{@:}\jv{+.}\ \jv{=}\ja{/}\bc\ \jn{fxy}$
  $\jn{mxy}\jo{=.}\bo\jv=\ja/\ \jv{+.}\jc"\jn1\ \jv{\~{}:}\jc{\&}\jn1\bc\ \jn{fxy}$
  $\jn{ranks}\jo{=.}\bo\jv{+.}\jc{\&}\bo\jn{2}\jc{\&}\bo\jv>\ja/\ja{\textbackslash}\bc\bc\ja/\jc{@}\bo\jn{1}\jc{\&}\jv{,.}\bc\ \jv{\#}\ \jv{|:}\jc{@}\bo\jv-\ja{\~{}}\ \jv+\ja/\ja{\textbackslash}\jc"\jn{1}\bc\bc\ \jn{mxy}$
  $\jn{ranks}\jo{=.}\bo\jn{l}\jv,\jn{r}\bc\ \jv{-.}\ja{\~{}}\ \jn{ranks}\ \bo\jv{]}\ \jv,\ \jv+\jc"\jn{1}\bc\ \jn{csx}\ \jv,\jc{\&}\jv{\#}\ \jn{csy}$
  $\jn{repr}\jo{=.}\jv{u}\jc{`}\jn{''}\ \jv{<}\jc{F..}\bo\jn{\textquotesingle\textquotedbl\textquotesingle}\ \jv{,}\jc{\&}\jv{<}\ \bo\jv{,}\jc{\&}\jv{<}\ \jn{'0'}\jc{\&}\jv{;}\bc\ja{\~{}}\bc\ \jn{ranks}$
  $\jn{'fx fy'}\jo{=.}\jn{mxy}\ \jv{<}\jc{@}\jv{|.}\jc{@}\jv{\#}\jc{"}\jn{1}\ \jn{fxy}$
  $\bo\jn{x}\bo\jv{\$}\jv{,}\bc\ja{\~{}}\jn{fx}\jv,\jn{csx}\bc\ \jn{repr}\jc{`:}\jn{6}\ \jn{y}\bo\jv{\$}\jv{,}\bc\ja{\~{}}\jn{fy}\jv,\jn{csy}$
$\jt{\}\}}$
\end{lstlisting}
\end{minipage}
\end{center}

For example, \jv[r]{\jv{+}\ja{Broadcastly}} applied to arrays of shapes \jn{1 4 1 1 6} and \jn{1 2 3 1 1 5 6} returns an array of shape \jn{1 2 3 4 1 5 6}. The shape and elements of the result are the same as those obtained in NumPy when applying the \texttt{+} operator to such arrays.

The implementation is included in the \texttt{api/python3} repository, which provides a Python package for interoperability with J \cite{pythonapi}. Additional examples demonstrating the application of the adverb in J can be run directly online in Juno IDE (\url{https://jsoftware.github.io/juno/app/?s=broadcastly}).

\section{Conclusions}
Within this work, the relationship between verb rank in J and broadcasting in NumPy was analyzed, and NumPy broadcasting was expressed in terms of J rank, resulting in an implementation of operator (adverb) \ja{Broadcastly}. The implemented adverb can also be useful in practice. Through the API available as a J add-on \cite{pythonapi}, J code can be called directly from Python, and arrays can be converted between J and NumPy in both directions. This makes it convenient to combine code written in NumPy and J within a single project. $\ja{Broadcastly}$ further simplifies interoperability by allowing both approaches to multidimensional array processing to be used in J and simplifies the translation of NumPy code into J. As a result, users are no longer constrained by an earlier choice of language or by familiarity with only one approach to multidimensional array processing.

\section{Acknowledgments}
The author would like to thank Jacek Chrząszcz (University of Warsaw) and Henry Rich (Jsoftware Inc.) for their helpful comments and discussions.

\bibliographystyle{unsrt}
\bibliography{main}

\end{document}